\def\llncs{0}       % 1 = using llncs
\def\fullversion{1} % 0 = short version; 1 = full version
\def\authnotes{0}   % 0 = authnotes off; 1 = authnotes on

\documentclass[10pt, journal]{IEEEtran}
\usepackage{url}
\usepackage{latexsym}
\usepackage{amscd,amsmath}
\usepackage{amsfonts}
\usepackage{amssymb}
\usepackage{color}
\usepackage{float}
\usepackage{longtable}
\usepackage{graphicx}
\usepackage{xspace}

\ifnum\llncs=0

\newtheorem{thm}{Theorem}[section]
\newtheorem{lem}[thm]{Lemma}
\newtheorem{cor}[thm]{Corollary}
\newtheorem{propo}[thm]{Proposition}
\newtheorem{clm}[thm]{Claim}
\newtheorem{defn}[thm]{Definition}
\newtheorem{assm}[thm]{Assumption}
\newtheorem{rem}[thm]{Remark}
\newtheorem{obs}[thm]{Observation}
\newtheorem{egs}[thm]{Example}
\newtheorem{expr}{Experiment}
\newtheorem{fct}[thm]{Fact}
\newtheorem{cons}[thm]{Construction}
\newtheorem{nte}[thm]{Note}

\newenvironment{theorem}{\begin{thm}}{\end{thm}}
\newenvironment{lemma}{\begin{lem}}{\end{lem}}
\newenvironment{corollary}{\begin{cor}}{\end{cor}}
\newenvironment{proposition}{\begin{propo}}{\end{propo}}
\newenvironment{definition}{\begin{defn}}{\end{defn}}
\newenvironment{assumption}{\begin{assm}\begin{em}}{\end{em}\end{assm}}
\newenvironment{claim}{\begin{clm}\begin{rm}}{\end{rm}\end{clm}}
\newenvironment{remark}{\begin{rem}\begin{em}}{\end{em}\end{rem}}

\newenvironment{fact}{\begin{fct}\begin{em}}{\end{em}\end{fct}}
\newenvironment{construction}{\begin{cons}\begin{rm}}{\end{rm}\end{cons}}

\else

\ifnum\fullversion=1
\newtheorem{thm}{Theorem}
\newtheorem{lem}[thm]{Lemma}
\newtheorem{cor}[thm]{Corollary}
\newtheorem{propo}[thm]{Proposition}
\newtheorem{clm}[thm]{Claim}
\newtheorem{defn}[thm]{Definition}
\newtheorem{assm}[thm]{Assumption}
\newtheorem{rem}[thm]{Remark}
\newtheorem{fct}[thm]{Fact}

\newtheorem{cons}[thm]{Construction}

\renewenvironment{theorem}{\begin{thm}}{\end{thm}}
\renewenvironment{lemma}{\begin{lem}}{\end{lem}}

\renewenvironment{definition}{\begin{defn}\begin{rm}}{\end{rm}\end{defn}}

\fi

\fi

\newcommand{\lemref}[1]{Lemma~\ref{#1}}

\ifnum\llncs=1
  
\else
  
\fi

\renewcommand{\eqref}[1]{\mbox{Equation~(\ref{#1})}}

\newlength{\saveparindent}
\newlength{\saveparskip}
\ifnum\llncs=0
\newcount\proofqeded
\newcount\proofended
\def\qed{{\hspace{1pt}\rule[-1pt]{3pt}{9pt}}
\end{rm}\addtolength{\parskip}{-0pt}
\setlength{\parindent}{\saveparindent}
\global\advance\proofqeded by 1 }
\def\qedenv{
\end{rm}\addtolength{\parskip}{-0pt}
\setlength{\parindent}{\saveparindent}
\global\advance\proofqeded by 1 }
\newenvironment{proof}%
 {\proofstart}%
 {\ifnum\proofqeded=\proofended~\qed\fi \global\advance\proofended by 1
  \medskip}
 {\proofenvstart}%
 {\ifnum\proofqeded=\proofended\qedenv\fi \global\advance\proofended by 1
  \medskip}
\makeatletter
\def\proofstart{\@ifnextchar[{\@oprf}{\@nprf}}
\def\proofenvstart{\@ifnextchar[{\@osprf}{\@nsprf}}
\def\@oprf[#1]{\begin{rm}\protect\vspace{6pt}\noindent{\bf Proof of #1:\ }%
\addtolength{\parskip}{5pt}\setlength{\parindent}{0pt}}
\def\@osprf[#1]{\begin{rm}\protect\vspace{6pt}\noindent
\addtolength{\parskip}{5pt}\setlength{\parindent}{0pt}}
\def\@nprf{\begin{rm}\protect\vspace{6pt}\noindent{\bf Proof:\ }%
\addtolength{\parskip}{5pt}\setlength{\parindent}{0pt}}
\def\@nsprf{\begin{rm}\protect\vspace{6pt}\noindent%
\addtolength{\parskip}{5pt}\setlength{\parindent}{0pt}}
\fi

\newcounter{ctr}

\newcounter{ectr}

\newlength{\savejot}
\ifnum\llncs=0

\newenvironment{newmath}{\begin{displaymath}%
\setlength{\abovedisplayskip}{6pt}%
\setlength{\belowdisplayskip}{6pt}%
\setlength{\abovedisplayshortskip}{8pt}%
\setlength{\belowdisplayshortskip}{8pt} }{\end{displaymath}}

\newenvironment{neweqnarrays}{\begin{eqnarray*}%
\setlength{\abovedisplayskip}{-4pt}%
\setlength{\belowdisplayskip}{-4pt}%
\setlength{\abovedisplayshortskip}{-4pt}%
\setlength{\belowdisplayshortskip}{-4pt}%
\setlength{\jot}{-0.4in} }{\end{eqnarray*}}

\newenvironment{newequation}{\begin{equation}%
\setlength{\abovedisplayskip}{6pt}%
\setlength{\belowdisplayskip}{6pt}%
\setlength{\abovedisplayshortskip}{8pt}%
\setlength{\belowdisplayshortskip}{8pt} }{\end{equation}}

\else

\fi

\newcommand{\authnote}[2]{
\ifnum\authnotes=1 
  \begin{center}
    \fbox{\begin{minipage}{5.7in}
      \textbf{#1 says:} #2
    \end{minipage}}
  \end{center} 
\fi
}

\ifnum\llncs=1

\else

\fi

\newcommand{\calA}{{\cal A}}
\newcommand{\calB}{{\cal B}}
\newcommand{\calC}{{\cal C}}

\newcommand{\calG}{{\cal G}}

\newcommand{\calV}{{\cal V}}

\newcommand{\calX}{{\cal X}}
\newcommand{\calY}{{\cal Y}}
\newcommand{\calZ}{{\cal Z}}

\def\bits{\{0,1\}}

\newcommand{\Prob}[1]{{\Pr\left[\,{#1}\,\right]}}

\newcommand{\ra}{\rightarrow}
\newcommand{\eps}{\varepsilon}
\newcommand{\Ext}{\mathsf{Ext}}
\newcommand{\SD}[1]{\mathsf{SD}\left({#1}\right)}

\begin{document}

% Title, authors, abstract and keywords are in the abstract file
\title{On the Commitment Capacity of\\ Unfair Noisy Channels}

\author{Claude Cr\'epeau, Rafael~Dowsley and Anderson~C.~A.~Nascimento 
\thanks{Claude Cr\'epeau is with the School of Computer Science, McGill University, Local 318, Pavillon McConnell, 3480 rue University Montr\'eal (Qu\'ebec), Canada H3A 2A7. Email: crepeau@cs.mcgill.ca}
\thanks{Rafael Dowsley is with the Department of Computer Science, Bar-Ilan University, Israel. Email: rafael@dowsley.net. He is supported by the BIU Center for Research in Applied Cryptography and Cyber Security in conjunction with the Israel National Cyber Bureau in the Prime Minister’s Office.}
\thanks{Anderson~C.~A.~Nascimento is with the Institute of Technology, University of Washington Tacoma, 1900 Commerce Street, Tacoma, WA 98402-3100, USA. E-mail: andclay@uw.edu.}
}

\maketitle

\begin{abstract}
Noisy channels are a valuable resource from a cryptographic point of view. 
They can be used for exchanging secret-keys as well as realizing other cryptographic 
primitives such as commitment and oblivious transfer. To be really useful, noisy channels have to be consider in the scenario where a cheating party has some degree of control over the channel characteristics. Damg\r{a}rd et 
al. (EUROCRYPT 1999) proposed a more realistic model where such level of control is permited to an adversary, the so called unfair noisy channels, and proved that they
can be used to obtain commitment and oblivious transfer protocols. Given that 
noisy channels are a precious resource for cryptographic purposes, one important question is 
determining the optimal rate in which they can be used. The commitment capacity has already been 
determined for the cases of discrete memoryless channels and Gaussian channels. In this work we 
address the problem of determining the commitment capacity of unfair noisy channels. We compute a single-letter characterization of the commitment capacity of unfair noisy channels. In the case where an adversary has no control over the channel (the fair case) our capacity reduces to the well-known capacity of a discrete memoryless binary symmetric channel. 
\end{abstract}


\IEEEpeerreviewmaketitle

\section{Introduction}

\textbf{Commitment protocols.} Consider the case of first-price sealed-bid auctions: 
the participants place their bids in sealed 
envelopes that are later on opened to determine the winner 
and how much he will pay. The sealed envelopes play a crucial 
role in this protocol since they protect the secrecy of
each bid during the bidding process 
but at the same time they preclude the winner from changing 
his bid after the first phase. In the digital world, commitment 
protocols have a role similar to that of the sealed envelopes.
They are cryptographic protocols involving two mutually distrustful 
parties, Alice and Bob. The idea of the protocol is very simple: 
in a first phase Alice commits to a value $c$ by sending some information 
to Bob which commits her to $c$ without revealing it. Later on, Alice can execute 
with Bob a second phase of the protocol in order to reveal the value 
$c$. From Alice's point of view the protocol should guarantee
that after the first phase no information about $c$ is leaked to Bob.
From Bob's point of view the protocol should guarantee that
Alice cannot change her mind after the first phase, i.e., there
is at most one value that Alice can successfully open
in an eventual execution of the second phase.

Commitment was introduced by Blum~\cite{Blum1983} and 
is one of the most fundamental cryptographic protocols, widely 
used as sub-protocol in applications such as secure 
multi-party computation~\cite{STOC:GolMicWig87,C:ChaDamVan87,STOC:ChaCreDam88}, contract 
signing~\cite{EveGolLem85} and zero-knowledge proofs~\cite{GolMicWig91,Goldreich01,BraChaCre88}. 
In this work we are interested in commitment protocols that are 
information-theoretically secure, i.e., both security guarantees 
should hold even against (computationally) unbounded adversaries.

\textbf{Noise-based cryptography.}
The potential of noisy channels for cryptographic purposes
was first noticed by Wyner~\cite{Wyner75} who proposed a scheme for 
exchanging a secret-key in the presence of an eavesdropper (henceforth 
denote Eve). Wyner considered a model in which 
Eve receives the transmitted symbols over 
a degraded channel with respect to the legitimate
receiver�'s channel. This possibility result was later extended to 
the class of general (non-degraded) broadcast channels
by Csisz\'{a}r and K\"{o}rner~\cite{IEEEIT:CsiKor78}.
Both models did not consider public communication.
Maurer~\cite{IEEEIT:Maurer93} proved that public communication can 
improve the parties' ability of generating a secret.
Ahlswede and Csisz\'{a}r~\cite{IEEEIT-AhlCsi93} also improved the previous
results. 

In the case of commitment protocols, the first solution based on 
noisy channels was developed by Cr\'{e}peau and 
Kilian~\cite{FOCS:CreKil88}. The efficiency of the commitment protocols
based on noisy channels was later improved by
Cr\'{e}peau~\cite{EC:Crepeau97}. 

In all these pioneering works, the case where an adversary can control the characteristics of the channel was not considered.

\textbf{Unfair noisy channels.}
Damg\r{a}rd et 
al.~\cite{EC:DamKilSal99} proposed a more realistic model, called 
unfair noisy channels (UNC), in which the error
probability of the channel is not exactly known by the protocol
participants and can be influenced by malicious parties. The honest parties only know that the crossover probability 
is between $\gamma$ and $\delta$ (for $0<\gamma<\delta<1/2$), and an adversary can set the crossover probability to any value in this range.
Damg\r{a}rd et al.~\cite{EC:DamKilSal99} proved that using UNC with certain parameters it is possible to implement an information-theoretically 
secure commitment protocol. 

Recently a variant of UNC known as elastic channel has been studied. On one hand, it has two restrictions with relation to UNC: (1) only a corrupt receiver can manipulate the
crossover probability to any value in the range $[\gamma;\delta]$; (2) when both parties are honest the crossover probability is always $\delta$. On the other hand, commitment (and even oblivious transfer) can be 
obtained for all parameters  $0<\gamma<\delta<1/2$ \cite{EC:KhuMajSah16,cryptoeprint:2016:120}. 

\textbf{Commitment capacity.}
Since noisy channels are valuable resources for cryptography, 
an important question is determining the optimal rate
in which they can be used to implement the 
diverse cryptographic primitives. In the case of commitment this amounts to 
determining the optimal ratio between the length of the committed
values and the number of uses of the noisy channel 
(i.e., the commitment capacity of the channel). This problem
was first raised by Winter et al.~\cite{WinNasIma03}, who also answered
the question for the case of discrete memoryless channels.
After that, the commitment capacity of Gaussian channels
was obtained by Nascimento et al.~\cite{IEEEIT:NBSI08}.
The question of determining the optimal way of using 
noisy channels was also studied for other cryptographic tasks, 
for instance in the vast literature on secrecy 
capacity~\cite{Wyner75,IEEEIT:CsiKor78,IEEEIT:LeuHel78,EC:OzaWyn84,IEEEIT:Maurer93,IEEEIT-AhlCsi93,IEEEIT:CsiNar04,ISIT:ParBla05,ISIT:BarRod06,LiYatTra10,ISIT:LiaPooSha07,IEEEIT:GopLaiElg08,LaiElgPoo07,IEEEIT:CsiNar08,IEEEIT:AFJK09,BagMotKha09,IEEEIT:EkrUlu11,IEEEIT:OggHas11}
and also in the works that deal with the oblivious transfer 
capacity~\cite{IEEEIT:NasWin08,ISIT:ImaMorNas06,ISIT:AhlCsi07,IEEEIT:PDMN11,DowNas14}.

Our contribution in this work is determining the commitment 
capacity of unfair noisy channels.
The work is organized as follows. In section~\ref{sec:pre}
we establish our notation and present some existing results
that will be used thereafter. In section~\ref{sec:unc}, 
we introduce the model and the security definitions, and also 
state our main theorem. In section~\ref{sec:dir} we prove the
direct part of the theorem and in section~\ref{sec:con} 
the converse. We finish with some concluding remarks in
section~\ref{sec:conc}.

\section{Preliminaries}\label{sec:pre}

\subsection{Notation}

We will denote by calligraphic letters the domains of random 
variables, by upper case letters the random variables and by 
lower case letters the realizations of the random variables. 
Other sets are also denoted by calligraphic letters and the 
cardinality of a set $\calX$ is denoted by $|\calX|$. For a random 
variable $X$ over $\calX$, we denote its probability distribution 
by $P_X: \calX \ra [0,1]$ with $\sum_{x \in \calX} P_X(x) =1$. 
For a joint probability distribution $P_{XY}: \calX \times \calY \ra [0,1]$, 
let $P_X(x) := \sum_{y \in \calY}P_{XY}(x,y)$ denote the marginal 
probability distribution and let $P_{X|Y}(x|y):=\frac{P_{XY}(x,y)}{P_Y(y)}$ 
denote the conditional probability distribution if $P_Y(y) \neq 0$. 
We write $U_n$ for a random variable uniformly distributed over $\bits^n$. 

If an algorithm (or a function) $f$ is randomized, we denote by 
$f(x;r)$ the result of computing $f$ on input $x$ with randomness 
$r$. If $a$ and $b$ are two strings of bits of the same length, 
we denote by $a \oplus b$ their bitwise XOR and by 
$\mathsf{HD}(a,b)$ their Hamming distance. The logarithm used in this paper are 
in base 2. 

\subsection{Entropies}

The binary entropy function of $x$ is denoted by 
$h(x):= -x \log x - (1-x)\log(1-x)$. For finite alphabets $\calX,\calY,\calZ$ and random variables
$X \in \calX,Y \in \calY, Z\in \calZ$, the entropy of $X$ is denoted by $H(X)$ and 
the mutual information between $X$ and $Y$ by $I(X;Y)$. The min-entropy is given by
\[
H_\infty(X)=\min\limits_{x}\log(1/P_{X}(x)).
\]
Its conditional version is defined as
\[
H_\infty(X|Y)=\min\limits_{y}H_\infty(X|Y=y). 
\]
The max-entropy is defined as 
\[
H_0(X)=\log|\{x\in X|P_{X}(x)>0\}|
\]
and its conditional version is given by
\[
H_0(X|Y)=\max\limits_{y}H_0(X|Y=y).
\]
The statistical distance between two probability distributions 
$P_X$ and $P_Y$ over the same domain $\calV$ is
\[
\SD{P_X,P_Y} := \frac{1}{2} \sum_{v \in \calV} |P_X(v) - P_Y(v)|.
\]
For $0 \leq \eps < 1$ the $\eps$-smooth versions 
of the above entropies are defined as
\[
H_\infty^\eps(X)=\max\limits_{X':\SD{P_{X'},P_{X}}\leq\eps}H_\infty(X'),
\]
\[
H_\infty^\eps(X|Y)=\max\limits_{X'Y':\SD{P_{X'Y'},P_{XY}}\leq\eps}H_\infty(X'|Y'),
\]
\[
H_0^\eps(X)=\min\limits_{X':\SD{P_{X'},P_{X}}\leq\eps}H_0(X'),
\]
\[
H_0^\eps(X|Y)=\min\limits_{X'Y':\SD{P_{X'Y'},P_{XY}}\leq\eps}H_0(X'|Y').
\]

We will need the chain rules for smooth entropies conditioned
on an additional random variable $Z$~\cite{AC:RenWol05}:
\begin{eqnarray*}
& &H_\infty^{\eps + \eps'}(XY|Z) - H_\infty^{\eps'}(Y|Z) \geq  H_\infty^{\eps}(X|YZ)\\
 & & \quad \geq  H_\infty^{\eps_1}(XY|Z) -H_0^{\eps_2}(Y|Z) - \log(1/(\eps - \eps_1 - \eps_2))
\end{eqnarray*}
and
\begin{eqnarray*}
& & H_0^{\eps + \eps'}(XY|Z) - H_0^{\eps'}(Y|Z)  \leq  H_0^{\eps}(X|YZ) \\
 & & \quad \leq H_0^{\eps_1}(XY|Z) - H_\infty^{\eps_2}(Y|Z) + \log(1/(\eps - \eps_1 - \eps_2)).
\end{eqnarray*}

\subsection{Strong Extractors and the Leftover-Hash Lemma}

In the direct part of our proof we will use strong randomness extractors,
therefore we present here the relevant definitions and properties.

\begin{definition}[Strong Randomness Extractors~\cite{NisZuc96,DORS08}]
Let $\Ext:  \bits^n \ra \bits^\ell$ be a probabilistic polynomial time function 
which uses $r$ bits of randomness. We say that $\Ext$ is an efficient 
$(n,m,\ell,\eps)\mathrm{-strong}$  $\mathrm{extractor}$ if for all probability 
distributions $P_W$ with $\mathcal{W} = \bits^n$ and such that 
$H_\infty(W) \geq m$, for random variables $R$ uniformly distributed in the bit-strings of length $r$ and 
$L$ uniformly distributed in the bit-strings of length $\ell$,
we have that $\SD{P_{\Ext(W;R), R},P_{L,R}} \leq \eps$. 
\end{definition}

Strong extractors can extract at most $\ell=m-2\log(\eps^{-1})+O(1)$ 
bits of nearly random bits~\cite{RadTas00} and this optimal bound is 
achieved by universal hash functions~\cite{CarWeg77}.

\begin{definition}[$t$-Universal Hash Functions]
A class $\calG$ of functions $\calA \ra \calB$ is 
$t$-\emph{universal} if, for any distinct $x_1, \ldots, x_t \in \calA$, 
when $g$ is chosen uniformly at random from $\calG$, the induced distribution on 
$(g(x_1), \ldots, g(x_t))$ is uniform over $\calB^t$.
\end{definition}

The leftover-hash lemma (similarly the privacy-amplification 
lemma)~\cite{STOC:ImpLevLub89,HILL99,BenBraRob88,IEEEIT:BBCM95,DORS08} guarantees 
that universal hash functions can be used to extract
$\ell=m-2\log(\eps^{-1})+2$ nearly random bits.

\begin{lemma}[Leftover-hash lemma]
Assume that a class $\calG$ of functions $G: \bits^n \ra \bits^\ell$ 
is 2-universal. Then for $G$ selected uniformly at random from 
$\calG$ we have that
\[
\SD{P_{G(W),G},P_{U_\ell,G}} \leq \frac{1}{2} \sqrt{2^{-H_\infty(W)}2^\ell}.
\]
In particular, universal hash functions are 
$(n,m,\ell,\eps)\mathrm{-strong}$  $\mathrm{extractors}$ 
whenever $\ell \leq m-2\log(\eps^{-1})+2$.
\end{lemma}

The following lemma by Rompel will be also useful.

\begin{lemma}[\cite{Rompel90}]\label{lem:rom}
Suppose $t$ is a positive even integer, $X_1, \cdots, X_u$ are $t$-wise independent random variables taking values in the range $[0,1]$, $X=\sum_{i=1}^u X_i$, $\mu = E[X]$, and $A>0$. Then
$$\Prob{|X-\mu|>A}<O\left(\left(\frac{t\mu + t^2}{A^2}\right)^{t/2}\right).$$
\end{lemma}

\section{Problem Statement}\label{sec:unc}

\subsection{Unfair Noisy Channels}

Since commitment protocols that are information-theoretically
secure cannot be implemented from scratch, the efforts
were focused on obtaining such protocols based on physical
assumptions. One of these assumptions is the existence
of noisy channels that can be used by the parties. Binary 
symmetric channels are known to allow the implementation of 
commitment schemes~\cite{FOCS:CreKil88,EC:Crepeau97}. 
But these protocols have the disadvantage that they rely on the assumption that a malicious party do not interfere with the channel to try to modify 
its error probability.

%But as argued 
%by Damg\r{a}rd et al.~\cite{EC:DamKilSal99} the protocols 
%based on noisy channels have the disadvantage that the participants rely on
%the exact knowledge of the channel's error probability, which in general is
%not realistic to assume since a malicious party can try to interfere with
%the channel and modify the error probability (and even if both
%parties are honest, it is difficult to estimate the error probability
%with arbitrary precision).

Damg\r{a}rd et al.~\cite{EC:DamKilSal99} 
introduced a more realistic assumption, namely unfair noisy 
channels, that is a modification of the binary symmetric channel. 
In this model, a channel is specified by two parameters: $\gamma$ 
and $\delta$ (with $0 < \gamma < \delta < \frac{1}{2}$); 
and the channel is denoted as $(\gamma, \delta)$-UNC.
The error probability of  the channel is guaranteed to fall into the 
interval $[\gamma,\delta]$, but is not known by the honest parties. 
Therefore any protocol based on a $(\gamma, \delta)$-UNC
should work for any error probability in the range $[\gamma,\delta]$.
On the other hand, a malicious party can set the error probability
to any value in the range $[\gamma,\delta]$. 

\begin{definition}[Unfair Noisy Channels] The $(\gamma, \delta)$-UNC
receives as input a bit $X$ and outputs a bit $Y$. The transition probability of the 
$(\gamma, \delta)$-UNC is determined by an auxiliary parameter $T$
whose alphabet are the real numbers in the interval $[\gamma,\delta]$. 
If the transmitter or the receiver is malicious, he can choose the value of 
$T$; otherwise it is randomly chosen and is not revealed to the parties.
The transition probability is given by $P_{Y|XT}(y,x,t)=1-t$ if $y=x$ and $P_{Y|XS}(y,x,t)=t$ if $y \neq x$.
\end{definition}

Damg\r{a}rd et al.~\cite{EC:DamKilSal99} proved that, on one hand, if 
$\delta \geq 2\gamma(1-\gamma)$ then the $(\gamma, \delta)$-UNC 
is trivial (i.e., secure commitment protocols cannot be 
based on this channel). On the other hand, if 
$\delta < 2\gamma(1-\gamma)$ then there is a commitment protocol 
based on the $(\gamma, \delta)$-UNC.

\emph{Remark:} Note that a $(\gamma, \delta)$-UNC can equivalently
be seen as the concatenation of two binary symmetric channels,
$W_f$ with error probability $\gamma$ and $W_v$ with
error probability  $\theta$ for  
$0 \leq \theta \leq \frac{\delta - \gamma}{1 - 2\gamma}$.
The error probability of the channel $W_v$ can be controlled 
by a malicious party and it is unknown in the case that both parties 
are honest.

\begin{definition} We say that two strings $x^n$ and $y^n$ are $\epsilon$-\emph{compatible} with an $(\gamma, \delta)$-Unfair Noisy Channel if, for $\epsilon >0$, the Hamming distance ($\mathsf{HD}$) of $x^n$ and $y^n$ is at most $\delta n +\epsilon$. Similarly, two random variables $X^n$ and $Y^n$ are said to be compatible with an $(\gamma, \delta)$-Unfair Noisy Channel if $\Prob{\mathsf{HD}(X^n,Y^n)>\delta n}$ is negligible in $n$. 
\end{definition}

\subsection{Commitment Capacity}

Since $(\gamma, \delta)$-UNC are valuable resources, one would 
like to use them in the most efficient way. Hence the important 
question of  determining the commitment capacity of these channels arises.
Our goal in this work is to determine the commitment capacity 
of unfair noisy channels in the same way that Winter et al.~\cite{WinNasIma03} 
did for the discrete memoryless channels and Nascimento et al.~\cite{IEEEIT:NBSI08} 
did for the Gaussian channels. Let us begin by defining the concept
of commitment protocols and recalling its security notions.

A commitment protocol has two phases: called commitment and opening. 
There are two parties involved in the protocol: the committer (also denoted
Alice) and the verifier (also denoted Bob). The protocol
works as follows. In the commitment phase, Alice commits to 
a value $c$, but without revealing anything about it to Bob. Later on,
Alice can execute the opening phase to disclose the value $c$
to Bob. The security guarantee that Alice expects from the protocol
is that nothing about $c$ should be learned by Bob in the first phase,
while the security guarantee that Bob expects is that Alice should
not be able to change the value committed to after the first phase.
We proceed with a more detailed description of these definitions
and of the resources available to the parties in our model.

Alice and Bob have two channels available between them:
\begin{itemize}
\item a bidirectional authenticated noiseless channel, and
\item $(\gamma, \delta)$-UNC from Alice to Bob
\end{itemize}
Note that this model allows multiple uses of the 
$(\gamma, \delta)$-UNC within a protocol in an
interactive manner. Let $n$ denote the number of times 
that the parties use the $(\gamma, \delta)$-UNC channel. 

\paragraph{Commitment Protocol}
A commitment protocol is a family of protocols indexed by the 
security parameter $n$. Each protocol uses the $(\gamma, \delta)$-UNC 
$n$ times and proceeds in two phases as described below. 
For the sake of simplicity in the notation we will not explicitly 
mention the dependence on the security parameter.
Both parties have access to local randomness. Note that all the 
messages generated by Alice and Bob are well-defined random variables, 
depending on the value that Alice wants to commit to, $c$, and the local randomness of 
the parties. As usual, we 
assume that the noiseless messages exchanged by the players and their 
personal randomness are taken from $\bits^*$.

\paragraph{Commitment Phase} Alice has an input $c$
(from the message space  $\calC$) that 
she wants to commit to.  
There are $n$ rounds of communication through
the $(\gamma, \delta)$-UNC and in each of these rounds Alice inputs a 
symbol $x_i$ to the $(\gamma, \delta)$-UNC and an output $y_i$ is
delivered to Bob. Let $X^n$ be the random variable denoting the bitstring sent through
the $(\gamma, \delta)$-UNC and $Y^n$ the bitstring received through 
the $(\gamma, \delta)$-UNC. The parties can use the bidirectional authenticated 
noiseless channel at any time, with the messages possibly depending on $X^n$, 
$C$ and the local randomness. Let $M$ be the random variable denoting all the noiseless messages 
exchanged between the players, 

\paragraph{Opening Phase} The parties only exchange messages over 
the noiseless channel. Alice announces the value $\hat{c}$ and the bitstring $\hat{x}^n$ that she claims that she used during the 
first phase. Finally Bob executes a test 
$\beta(\hat{x}^n, y^{n}, m, \hat{c})$ in order to decide if he accepts the value $\hat{c}$ or not.
\\

We call Bob's view all the data in his possession after the completion of 
the commitment phase and denote it by $view_{B}$. 

A commitment protocol is $\epsilon$-\textit{sound} if for honest Alice and Bob 
executing the protocol and for all $c \in \calC$ and for any $X^n$ compatible with $Y^n$ for the $(\gamma, \delta)$-UNC,
\[
\Prob{\beta(X^n, Y^n, M, c) = accept} \geq 1 - \epsilon.
\]
A commitment protocol is $\epsilon$-\textit{concealing} if for any possible 
behavior of Bob in the commitment phase,
\[
I(C;View_B) \leq \epsilon.
\]
A protocol is $\epsilon$-\textit{binding}, if for all $\widetilde{c},\overline{c} \in \calC$ such that $\widetilde{c} \neq \overline{c}$
and for any strategy of (a dishonest) Alice to pick the random variable $X^n$ that is sent through 
the $(\gamma, \delta)$-UNC channel, and the random variables $\widetilde{X}^n$ and $\overline{X}^n$ 
that are presented during the opening phase
\begin{equation*}
\Prob{
\begin{array}{l}
\beta(\widetilde{X}^n,Y^n, M, \widetilde{c}) = accept~\&\\ 
\beta(\overline{X}^n,Y^n, M, \overline{c}) = accept
\end{array} 
} \leq \epsilon.
\end{equation*}

We call a commitment protocol \textit{secure} if there exists a function $\epsilon$
that is negligible in the security parameter $n$ and is such that the protocol
is  $\epsilon$-\textit{sound}, $\epsilon$-\textit{concealing} and 
$\epsilon$-\textit{binding}.\footnote{For easy of presentation the security of 
the constructions is argued in the stand-alone model (as usual in cryptography)
in which case there is only one execution of
the protocol. But the security of the commitment protocols based on noisy 
channels can be extended to the UC framework~\cite{FOCS:Canetti01} in which the protocols 
can be composed and arbitrary protocols can be executed in parallel~\cite{JIT:DGMN13,SBSEG:DMN08}.}

\emph{Remark:} We restrict our model to protocols where the public conversation $M$ does not depend on the channel output $Y^n$ given $X^n$, that is $I(M;Y^n|X^n)=0$. This is indeed the case for all the protocols in the literature. Moreover, the public communication is used solely to prevent Alice from cheating, thus we see no reason for a commitment protocol based on noisy channels to have its public communication depending on the channel output for a given input $X^n$. 

The \emph{commitment rate} of the protocol is given by
\[
R_C=\frac{\log|\calC|}{n}.
\]

A commitment rate is said to be achievable if there exists 
a secure commitment protocol that achieves this rate. The 
\emph{commitment capacity} of a $(\gamma, \delta)$-UNC 
is the supremum of the achievable rates.

Our main result is presented below and states the commitment 
capacity of the $(\gamma, \delta)$-UNC. The proof appears in 
sections~\ref{sec:dir} and~\ref{sec:con}.

\begin{theorem}
The commitment capacity of any non-trivial $(\gamma, \delta)$-Unfair 
Noisy Channel is given by
\[
h(\gamma)-h(\theta), \mathrm{~for~}\theta=\frac{\delta - \gamma}{1 - 2\gamma}.
\]

\end{theorem}

\section{Protocol - Direct Part}\label{sec:dir}

We first prove the direct part of the theorem,
i.e., we describe the protocol that achieves the commitment
capacity and prove its security. The protocol follows the approach of Damg\r{a}rd et 
al.~\cite{EC:DamKilSal99} and uses two rounds of hash challenge-response in order 
to guarantee the binding property: the intuition is that the first round reduces the 
number of inputs that Alice can successfully open to be polynomial in the security parameter.
The second round then binds Alice to one specific input.
The concealing condition is achieved using a 2-universal 
hash function $\Ext$ chosen by Alice that is used to generate 
a secure key which is then applied as a one-time pad to cipher $c$.

Let $\theta=\frac{\delta - \gamma}{1 - 2\gamma}$ and let
$\nu > 0$ be a parameter of the protocol. Let $\alpha_1$, $\alpha_2$, 
$\beta$ be parameters such that $\alpha_1,\alpha_2>0$, $\beta>\alpha_1+\alpha_2$, and 
$n(h(\theta)+\alpha_1)$,  $n\alpha_2$ and $n(h(\gamma)-h(\theta)-\beta)$ are integers.
In the following commitment protocol the message space 
is $\calC = \bits^{n(h(\gamma)-h(\theta)-\beta)}$.

Commitment Phase: Alice wants to commit to the binary string $c \in \calC$. The parties proceed as follows:
\begin{enumerate}
 	\item Alice chooses a random binary string $x^n = (x_1,\ldots, x_n)$ of dimension $n$ and for $1 \leq i \leq n$ sends the bit $x_i$ to Bob over the $(\gamma, \delta)$-UNC.
	\item Bob receives the string $y^n = (y_1,\ldots ,y_n)$ sent over the $(\gamma, \delta)$-UNC, chooses uniformly at random a function $g_1$ of the class of $4n$-universal hash functions $\calG_1:\bits^n \ra \bits^{n(h(\theta)+\alpha_1)}$, and sends the description of $g_1$ to Alice over the noiseless channel.
	\item Alice computes $e_1=g_1(x^n)$ and sends it to Bob.
	\item Bob chooses uniformly at random a function $g_2$ of the class of $2$-universal hash functions $\calG_2:\bits^n \ra \bits^{n\alpha_2}$, and sends its description to Alice over the noiseless channel.		
	\item Alice chooses uniformly at random a function $\Ext:\bits^n \ra \bits^{n(h(\gamma)-h(\theta)-\beta)}$ of the class of two-universal hash functions, computes $d = c \oplus \Ext(x^n)$ and $e_2=g_2(x^n)$, and sends $d$, $e_2$ and the description of $\Ext$ to Bob over the noiseless channel.
\end{enumerate}
		
Opening Phase: Alice wants to reveal the value of $c$ to Bob. The parties proceeds as follow:
\begin{enumerate}
	\item Alice sends to Bob over the noiseless channel the values $x^n$ and $c$.
	\item Bob checks if $\gamma n - \nu n \leq \mathsf{HD}(x^n, y^n) \leq \delta n + \nu n$, if $g_1(x^n)=e_1$, $g_2(x^n)=e_2$ and if $c =\Ext(x^n)\oplus d$. Bob accepts if there are no problems in the tests.
\end{enumerate}

\begin{theorem}
The protocol is secure and can achieve the commitment rate 
$h(\gamma)-h(\theta)- \beta$ for any $h(\gamma)-h(\theta) > \beta > 0$ and $n$ 
sufficiently large.
\end{theorem}
\begin{proof}We proof that the protocol is binding, concealing and 
sound, and furthermore achieves the desired commitment rate.

\paragraph{Soundness} The protocol fails for honest Alice and 
Bob only if $\mathsf{HD}(x^n, y^n) > \delta n + \nu n$ or 
$\mathsf{HD}(x^n, y^n) < \gamma n - \nu n$. We have that 
the expectation of $\mathsf{HD}(x^n, y^n)$ is less than 
or equal to $\delta n$, because the $(\gamma, \delta)$-UNC
 has error probability less than or equal to $\delta$. So the 
 Chernoff bound guarantees that the probability that  
 $\mathsf{HD}(x^n, y^n) > \delta n + \nu n$ is a negligible function of $n$. 
 Similarly we can prove that the probability that 
 $\mathsf{HD}(x^n, y^n) < \gamma n - \nu n$ is a negligible function of $n$.

\paragraph{Concealment} For any $\eta >0$ and $n$ sufficiently large, we have 
that
\begin{eqnarray*}
&& H_\infty^{\eps}(X^n|G_1(X^n),G_2(X^n),Y^n,G_1,G_2) \geq \\
&& \quad \geq H_\infty(X^n,G_1(X^n),G_2(X^n)|Y^n,G_1,G_2)\\
&& \quad \quad - H_0(G_1(X^n),G_2(X^n)|Y^n,G_1,G_2) - \log(\eps^{-1})\\
&& \quad =  H_\infty(X^n|Y^n,G_1,G_2) \\
&& \quad \quad - H_0(G_1(X^n),G_2(X^n)|Y^n,G_1,G_2) - \log(\eps^{-1})\\
&& \quad =  H_\infty(X^n|Y^n)\\
&& \quad \quad - H_0(G_1(X^n),G_2(X^n)|Y^n,G_1,G_2) - \log(\eps^{-1})\\
&& \quad \geq  n(h(\gamma)-\eta)-n(h(\theta)+\alpha_1+\alpha_2) - \log(\eps^{-1})\\
&& \quad =  n(h(\gamma)-h(\theta)-\eta-\alpha_1-\alpha_2) - \log(\eps^{-1})
\end{eqnarray*}
where the first inequality follows from the chain rule for smooth entropy, the first
equality from the fact that $G_1(X^n),G_2(X^n)$ are functions of $G_1$, $G_2$ and $X^n$, the second equality
from the fact that $X^n$ is independent of $G_1$, $G_2$ given $Y^n$ and the last inequality follows
from the facts that the error probability of the channel is at least $\gamma$, 
the range of $G_1$ has cardinality $2^{n(h(\theta)+\alpha_1)}$ and the range of $G_2$ has cardinality $2^{n\alpha_2}$.

Setting $\eps = 2^{-\kappa n}$ (with $\kappa >0$), 
for $n$ sufficiently large, 
the security of the key obtained by applying the hash 
function $\Ext:\bits^n \ra \bits^{n(h(\gamma)-h(\theta)-\beta)}$ to $x$
follows from the leftover-hash lemma since $\beta>\alpha_1+\alpha_2$
and $\kappa$ and $\eta$ can be arbitrarily small for $n$ sufficiently large.
Note that having a negligible statistical distance is equivalent to having a negligible mutual information \cite{IEEEIT:DamPedPfi98}.

\paragraph{Binding} 
A dishonest Alice can cheat successfully 
only if she finds two different strings $\overline{x}^n$ and $\widetilde{x}^n$ 
such that $\gamma n - \nu n \leq \mathsf{HD}(\overline{x}^n, y^n) \leq \delta n + \nu n$, 
$\gamma n - \nu n \leq \mathsf{HD}(\widetilde{x}^n, y^n) \leq \delta n + \nu n$, and both pass the sequentially performed 
hash challenge-response tests,
for arbitrarily small $\nu$ and sufficiently large $n$. We can assume without loss of 
generality that Alice sets the error probability of the channel to 
$\gamma$ when she sends $x^n$. In the typicality test an honest Bob 
accepts any string that is jointly typical with $y^n$ for any 
error probability $\gamma \leq \rho \leq \delta$. So Alice can cheat only if she finds two strings $\overline{x}^n$ and $\widetilde{x}^n$ 
so that both pass the hash tests and are jointly typical with $x^n$ for binary symmetric channels with error probabilities 
$0 \leq \overline{\tau} \leq \theta$ and $0 \leq \widetilde{\tau} \leq \theta$. The number of such jointly typical 
strings is upper bounded by $2^{n(h(\theta)+\eps')}$
for any $\eps'> 0$ and $n$ sufficiently large. We fix $\alpha_1 > \eps'$.

Let the viable set denote the channel inputs that Alice can possibly open to Bob and he would accept. If there were no hash checks, the viable set would have at most $2^{n(h(\theta)+\eps')}$ elements. Lets consider this initial viable set. The goal of the first round of hash challenge-response is to, with overwhelming probability, reduce the number of elements of the viable set to at most $8n+1$. In this first round, Alice has to commit to one arbitrary value $e_1$ for the output of the hash function $g_1$. Considering the $j$-th viable string before this first round, we define $I_j$ as $1$ if that string is mapped to $e_1$ by $g_1$; and $I_j=0$ otherwise. Let $I =\sum_j I_j$. Clearly $\mu=E[I]<1$ as $\alpha_1 > \eps'$. Let $e_1$ be considered bad if $I$ is bigger than $8n+1$. Given that $g_1$ is $4n$-wise independent, by applying \lemref{lem:rom} with $t=4n$ and $A=2t=8n$, we get
\begin{eqnarray*}
\Prob{I>8n+1}&<&O\left(\left(\frac{t\mu + t^2}{(2t)^2}\right)^{t/2}\right)\\
&<&O\left(\left(\frac{1 + t}{4t}\right)^{t/2}\right)\\
&<&O\left(2^{-t/2}\right).
\end{eqnarray*}

Then the probability that any $e_1$ is bad is upper bounded by 
$$O\left(2^{n(h(\theta)+\alpha_1)}2^{-t/2}\right) < O\left(2^{-n}\right).$$

But if the viable set contains at most $8n+1$ elements after the first hash challenge-response round, the probability that some of those collide 
in the second hash challenge-response round is upper bounded by
$$\left(8n+1\right)^22^{-\alpha_2 n},$$
which is negligible in $n$.

\paragraph{Commitment Rate} For $n$ sufficiently large, $\alpha_1$ and $\alpha_2$ can 
be made arbitrarily small, and thus $\beta$ can also be made arbitrarily small 
while preserving the security of the protocol. Therefore it is possible to achieve 
the commitment rate $h(\gamma)-h(\theta)- \beta$ for any $h(\gamma)-h(\theta) > \beta > 0$.
\end{proof}

\section{Converse}\label{sec:con}

%For proving the converse, we will assume a specific cheating behavior for Alice. As we are interested in proving an upper bound in the commitment capacity, restricting Alice's behavior will only strength our result. Let $k=\log |\mathcal{C}|$ and $C$ be uniformly random over $\mathcal{C}$. We assume Alice computes a random binary string $Z^n$ and then computes a random variable $X^n$ by passing $Z^n$ through a binary symmetric channel with error probability equal to $\theta$ with $0 \leq \theta \leq \frac{\delta-\gamma}{1-2\gamma}$. Alice sets the noise level of the unfair noisy channel connecting her to Bob to $\gamma$. Alice then uses $X^n$ as her input to the unfair noisy channel  in a commitment protocol. Denote by $Y^n$ the random variable representing the string received by Bob when $X^n$ is input in the unfair noisy channel. Denote the conversation over the public authenticated and noiseless channel by $M$. 
%We first give a procedure so that the commitment $C$ can be estimated with high probability from $Z^n,Y^n$ and $M$. 

For proving the converse, we will assume a specific cheating behavior for Alice. As we are interested in proving an upper bound in the commitment capacity, restricting Alice's behavior will only strength our result. Let $k=\log |\mathcal{C}|$ and $C$ be uniformly random over $\mathcal{C}$. Let $X^n$ be a random variable representing the data Alice inputs into the unfair noisy channel.  Assume, Alice sets the noise level of the unfair noisy channel connecting her to Bob to $\gamma$. Let $Y^n$ be a random variable obtained by passing $X^n$ through the unfair noisy channel (Channel 1). Let $Z^n$ be a random variable obtained by passing $X^n$ through a binary symmetric channel with error probability equal to $\theta$ with $0 \leq \theta \leq \frac{\delta-\gamma}{1-2\gamma}$ (Channel 2). Denote the conversation over the public authenticated and noiseless channel by $M$. 

In the case of commitments based on fair noisy channels, it was proved  in \cite{WinNasIma03} that after the commit phase is finished, if Bob is presented with Alice's inputs to the channel, $X^n$, he is able to obtain almost complete knowledge on the committed value $C$.  Here we will prove that in the case of unfair noisy channels if Bob is presented with a noisy version of $X^n$ he is still able to compute the committed value $C$ with high probability. 
\begin{lemma}
$H(C|Z^nM) \leq 1 + kp$ for $p$ negligible in $n$.
\end{lemma}
\begin{proof} 
Let $C, X^n, Y^n, Z^n$ and $M$ be defined as above.  
We first give a procedure so that the commitment $C$ can be estimated with high probability from $Z^n,Y^n$ and $M$.

The procedure is as follows: given $Z^n,Y^n$ and $M$, compute the value $c$ that maximizes $\Prob{\beta(Z^n, Y^n, M, c) = accept}$, breaking ties in an arbitrary way. Because of the bindingness condition, we know that no two different values of $c$ will be accepted with high probability

\begin{equation*}
\Prob{
\begin{array}{l}
\beta(\widetilde{X}^n,Y^n, M, \widetilde{c}) = accept~\&\\ 
\beta(\overline{X}^n,Y^n, M, \overline{c}) = accept
\end{array} 
} \leq \epsilon
\end{equation*}
for all $\widetilde{X}^n$ and $\overline{X}^n$ compatible with $Y^n$.

Moreover, from the correctness property of the protocol and from the fact that $Z^n$ and $Y^n$ are compatible for the unfair noisy channel in question, we know that for the correct value $c$ we have 

\[
\Prob{\beta(Z^n, Y^n, M, c) = accept} \geq 1 - \epsilon.
\]

Thus, with high probability this procedure will give us the right committed value $c$. Let $p$ be the probability that this procedure returns a wrong value.
Using Fano's inequality we get
\begin{eqnarray*}
H(C|Z^nY^nM) &\leq&h(p)+p \log |\mathcal{C}|\\
&\leq& 1 + p \log |\mathcal{C}|\\
&\leq & 1 + kp.
\end{eqnarray*}

One can prove that the output of the channel $Y^n$ is not needed in the above described procedure. The intuition is that $Z^n$ is a ``less noisy'' version of $X^n$ and thus can be used (instead of $Y^n$) for retrieving the value of the commitment. In order to formalize this result, we need to use the assumed independence of the public conversation $M$ and $Y^n$ given $X^n$, i.e., $I(M;Y^n|X^n)=0$. First, we pass (or simulate the passing of) $Z^n$ through a binary symmetric channel with error probability $\gamma$. Denote the output of the simulated channel by $\widetilde{Y}^n$. Note that $\widetilde{Y}^n$ and $X^n$ are compatible. Moreover, given the fact that the public conversation is independent of $Y^n$ given $X^n$, we have that the public communication of the protocol using Channel 1 ($M$), $C,X^n$ and $\widetilde{Y}^n$ are a valid transcript of a commitment protocol (the correctness, binding and concealing properties should apply). Thus, one has from the correctness property that 
\[
\Prob{\beta(Z^n, \widetilde{Y}^n, M, c) = accept} \geq 1 - \epsilon.
\]

From bindingness we have that 

\begin{equation*}
\Prob{
\begin{array}{l}
\beta(\widetilde{X}^n,\widetilde{Y}^n, M, \widetilde{c}) = accept~\&\\ 
\beta(\overline{X}^n,\widetilde{Y}^n, M, \overline{c}) = accept
\end{array} 
} \leq \epsilon
\end{equation*}
for all $\widetilde{X}^n$ and $\overline{X}^n$ compatible with $\widetilde{Y}^n$. 

Again, using  Fano's inequality we get
\begin{eqnarray*}
H(C|Z^n\widetilde{Y}^nM) &\leq&h(p)+p \log |\mathcal{C}|\\
&\leq& 1 + p \log |\mathcal{C}|\\
&\leq & 1 + kp.
\end{eqnarray*}

Because the Markov chain $CX^n \leftrightarrow Z^n\leftrightarrow \widetilde{Y}^n$ holds, we have that  $H(C|Z^n\widetilde{Y}^nM)=H(C|Z^nM)$, which proves our result.

\end{proof}

We have that
\begin{eqnarray}
k & \leq & H(C|Y^nM) + \epsilon \\
\nonumber & = & H(C|Y^nM) - H(C|Z^nM) + H(C|Z^nM) +  \epsilon \\
& \leq & H(C|Y^nM) - H(C|Z^nM) + 1 + kp +  \epsilon \\
\nonumber &=& H(C|Y^nM) - H(C|Z^nM) - H(C|M) + H(C|M)\\
\nonumber && \quad + 1 + kp +  \epsilon \\
\nonumber &=& I(C;Z^n|M) - I(C;Y^n|M)+ 1 + kp +  \epsilon
\end{eqnarray}
where inequality 1 comes from the $\epsilon$-concealing requirement and 
inequality 2 from the previous lemma. 

Now we develop the expression $I(C;Z^n|M) - I(C;Y^n|M)$ using the same steps used 
on the Section V of the seminal work of Csisz\'ar and K\"{o}rner \cite{IEEEIT:CsiKor78}; the 
details are included for the sake of completeness. Let $Z^i$ denote $Z_1 \ldots Z_i$ and $\hat{Y}^i$ denote $Y_i \ldots Y_n$. We expand 
$I(C;Z^n|M)$ starting from $I(C;Z_1|M)$ and $I(C;Y^n|M)$ starting from $I(C;Y_n|M)$
\begin{eqnarray}
\nonumber I(C;Z^n|M)&=&\sum_{i=1}^n I (C;Z_i|MZ^{i-1})\\
\nonumber &=& \sum_{i=1}^n \Big[ H(Z_i|MZ^{i-1})- H(Z_i|MZ^{i-1}C) \\
\nonumber  && \quad  - H(Z_i|MZ^{i-1}C\hat{Y}^{i+1})\\ 
\nonumber  && \quad  + H(Z_i|MZ^{i-1}C\hat{Y}^{i+1})\Big] \\
\nonumber &=& \sum_{i=1}^n \Big[I(C\hat{Y}^{i+1};Z_i|MZ^{i-1})\\
\nonumber  && \quad - I (\hat{Y}^{i+1};Z_i|MZ^{i-1}C) \Big]\\
\nonumber &=& \sum_{i=1}^n \Big[I(C;Z_i|MZ^{i-1}\hat{Y}^{i+1})\\
\nonumber  && \quad  + I(\hat{Y}^{i+1};Z_i|MZ^{i-1})\\
\nonumber  && \quad - I (\hat{Y}^{i+1};Z_i|MZ^{i-1}C) \Big].
\end{eqnarray}
Similarly we obtain
\begin{eqnarray}
\nonumber I(C;Y^n|M)&=& \sum_{i=1}^n \Big[I(C;Y_i|MZ^{i-1}\hat{Y}^{i+1})\\
\nonumber  && \quad  + I(Z^{i-1};Y_i|M\hat{Y}^{i+1})\\
\nonumber  && \quad - I (Z^{i-1};Y_i|M\hat{Y}^{i+1}C) \Big].
\end{eqnarray}
We have that 
\begin{eqnarray}
\nonumber \sum_{i=1}^n I(\hat{Y}^{i+1};Z_i|MZ^{i-1}) &=& \sum_{i=1}^n \sum_{j=i+1}^n I(Y_j;Z_i|MZ^{i-1}\hat{Y}^{j+1})\\
\nonumber &=& \sum_{j=2}^n \sum_{i=1}^{j-1} I(Y_j;Z_i|MZ^{i-1}\hat{Y}^{j+1})\\
\nonumber &=& \sum_{j=1}^n I(Z^{j-1};Y_j|M\hat{Y}^{j+1}).
\end{eqnarray}
Similarly we can get that
\begin{eqnarray}
\nonumber \sum_{i=1}^n I(\hat{Y}^{i+1};Z_i|MZ^{i-1}C)=\sum_{j=1}^n I(Z^{j-1};Y_j|M\hat{Y}^{j+1}C).
\end{eqnarray}
Therefore
\begin{eqnarray}
\nonumber I(C;Z^n|M) - I(C;Y^n|M)&=& \sum_{i=1}^n \Big[I(C;Z_i|MZ^{i-1}\hat{Y}^{i+1}) \\
\nonumber  && \quad - I(C;Y_i|MZ^{i-1}\hat{Y}^{i+1})\Big].
\end{eqnarray}
Letting $L$ be a random variable uniformly distributed in $\{1,\ldots, n\}$ and independent of $CMX^nY^nZ^n$, 
and setting $U\triangleq MZ^{L-1}\hat{Y}^{L+1}L$, $V\triangleq UC$, $X\triangleq X_L$, $Y \triangleq Y_L$ and 
$Z \triangleq Z_L$ we get that $U\leftrightarrow V\leftrightarrow X \leftrightarrow YZ$ form a Markov chain and
\begin{align}
\nonumber \frac{1}{n}  \sum_{i=1}^n \Big[ I(C;Z_i|MZ^{i-1}\hat{Y}^{i+1}) - I(C;Y_i|MZ^{i-1}\hat{Y}^{i+1})\Big]=\\
\nonumber \qquad = I(C;Z|U) - I(C;Y|U)\\
\nonumber \qquad = I(V;Z|U) - I(V;Y|U).
\end{align}
Putting everything together, for any secure commitment protocol, there are $U\leftrightarrow V\leftrightarrow X \leftrightarrow YZ$
such that 
\begin{eqnarray}
\frac{k}{n} \leq  I(V;Z|U) - I(V;Y|U) + \epsilon'
\end{eqnarray}
where $\epsilon' = \frac{1 + kp +  \epsilon}{n}$ goes to 0 for $n$ sufficiently large.

We now set $\theta=\frac{\delta - \gamma}{1 - 2\gamma}$.
In our case channel 2 is less noisy than channel 1, therefore maximizing over all 
$U\leftrightarrow V\leftrightarrow X \leftrightarrow YZ$ we get
\begin{eqnarray}
\nonumber I(V;Z|U) - I(V;Y|U) &=& I(V;Z) - I(V;Y)\\
\nonumber && \quad- [I(U;Z)-I(U;Y)]\\
\nonumber &=& I(X;Z) - I(X;Y)\\
\nonumber && \quad - [I(X;Z|V)-I(X;Y|V)]\\
\nonumber && \quad- [I(U;Z)-I(U;Y)]\\
&\leq& I(X;Z) - I(X;Y)\\
&\leq& h(\gamma)-h(\theta)
\end{eqnarray}
where inequality 4 comes from the fact that both 
expressions in the brackets are non-negative since channel 2 is less noisy than channel 1 
and inequality 5 follows taking the maximum over $X$. Hence 
\begin{eqnarray}
\nonumber \frac{k}{n} \leq h(\gamma)-h(\theta) + \epsilon'
\end{eqnarray}
where $\epsilon'$ goes to 0 for $n$ sufficiently large. This completes the proof
of the converse.

\section{Final Remarks}\label{sec:conc}
In this paper we obtained the commitment capacity of the 
unfair noisy channels. A similar notion of 
capacity for oblivious transfer was defined in~\cite{IEEEIT:NasWin08}. We state 
as an open problem to obtain the oblivious transfer capacity of unfair
noisy channels. Another open question is to derive the commitment 
capacity of weak channels~\cite{TCC:Wullschleger09}. In the case of elastic channels, for commitments from Alice to Bob, the channel is essentially degraded to a binary symmetric channel with crossover probability $\gamma$ and therefore the commitment capacity is $h(\gamma)$. On the other hand, we conjecture that the commitment capacity for commitments from Bob to Alice is $h(\delta) - h(\theta)$ for $\theta=\frac{\delta-\gamma}{1-2\gamma}$. However, if either of the two restrictions of the elastic channels in relation to unfair noisy channels is discarded (i.e., only the receiver being able to set the crossover probability; and the crossover probability being fixed to $\delta$ when both parties are honest), then we get back to the same commitment capacity as for the unfair noisy channels.

Finally, we need in our proof to assume that $I(M;Y^n|X^n)=0$, that is, the public communication and the channel outputs are independent given the channel inputs. While that is a natural condition respected by all the protocols that exist, we believe that, ultimately, a more general proof that does not depend on this restriction is possible. We leave such a proof as a future research problem.

% --- -----------------------------------------------------------------
% --- The Bibliography.
% --- -----------------------------------------------------------------

\bibliographystyle{IEEEtran}
\bibliography{abbrev0,crypto,informationtheory,additional}

\begin{thebibliography}{10}
\providecommand{\url}[1]{#1}
\csname url@samestyle\endcsname
\providecommand{\newblock}{\relax}
\providecommand{\bibinfo}[2]{#2}
\providecommand{\BIBentrySTDinterwordspacing}{\spaceskip=0pt\relax}
\providecommand{\BIBentryALTinterwordstretchfactor}{4}
\providecommand{\BIBentryALTinterwordspacing}{\spaceskip=\fontdimen2\font plus
\BIBentryALTinterwordstretchfactor\fontdimen3\font minus
  \fontdimen4\font\relax}
\providecommand{\BIBforeignlanguage}[2]{{%
\expandafter\ifx\csname l@#1\endcsname\relax
\typeout{** WARNING: IEEEtran.bst: No hyphenation pattern has been}%
\typeout{** loaded for the language `#1'. Using the pattern for}%
\typeout{** the default language instead.}%
\else
\language=\csname l@#1\endcsname
\fi
#2}}
\providecommand{\BIBdecl}{\relax}
\BIBdecl

\bibitem{Blum1983}
\BIBentryALTinterwordspacing
M.~Blum, ``Coin flipping by telephone a protocol for solving impossible
  problems,'' \emph{SIGACT News}, vol.~15, no.~1, pp. 23--27, Jan. 1983.
  [Online]. Available: \url{http://doi.acm.org/10.1145/1008908.1008911}
\BIBentrySTDinterwordspacing

\bibitem{STOC:GolMicWig87}
O.~Goldreich, S.~Micali, and A.~Wigderson, ``How to play any mental game or {A}
  completeness theorem for protocols with honest majority,'' in \emph{19th
  Annual {ACM} Symposium on Theory of Computing}, A.~Aho, Ed.\hskip 1em plus
  0.5em minus 0.4em\relax New York City, NY, USA: {ACM} Press, May~25--27,
  1987, pp. 218--229.

\bibitem{C:ChaDamVan87}
D.~Chaum, I.~Damg{\aa}rd, and J.~{van de Graaf}, ``Multiparty computations
  ensuring privacy of each party's input and correctness of the result,'' in
  \emph{Advances in Cryptology -- {CRYPTO}'87}, ser. Lecture Notes in Computer
  Science, C.~Pomerance, Ed., vol. 293.\hskip 1em plus 0.5em minus 0.4em\relax
  Santa Barbara, CA, USA: Springer, Heidelberg, Germany, Aug.~16--20, 1988, pp.
  87--119.

\bibitem{STOC:ChaCreDam88}
D.~Chaum, C.~Cr{\'e}peau, and I.~Damg{\r a}rd, ``Multiparty unconditionally
  secure protocols (extended abstract),'' in \emph{20th Annual {ACM} Symposium
  on Theory of Computing}.\hskip 1em plus 0.5em minus 0.4em\relax Chicago, IL,
  USA: {ACM} Press, May~2--4, 1988, pp. 11--19.

\bibitem{EveGolLem85}
\BIBentryALTinterwordspacing
S.~Even, O.~Goldreich, and A.~Lempel, ``A randomized protocol for signing
  contracts,'' \emph{Commun. ACM}, vol.~28, no.~6, pp. 637--647, Jun. 1985.
  [Online]. Available: \url{http://doi.acm.org/10.1145/3812.3818}
\BIBentrySTDinterwordspacing

\bibitem{GolMicWig91}
O.~Goldreich, S.~Micali, and A.~Wigderson, ``Proofs that yield nothing but
  their validity or all languages in {NP} have zero-knowledge proof systems,''
  \emph{Journal of the {ACM}}, vol.~38, no.~3, pp. 691--729, 1991.

\bibitem{Goldreich01}
O.~Goldreich, \emph{Foundations of Cryptography: Basic Tools}.\hskip 1em plus
  0.5em minus 0.4em\relax Cambridge, UK: Cambridge University Press, 2001,
  vol.~1.

\bibitem{BraChaCre88}
\BIBentryALTinterwordspacing
G.~Brassard, D.~Chaum, and C.~Cr{\'e}peau, ``Minimum disclosure proofs of
  knowledge,'' \emph{J. Comput. Syst. Sci.}, vol.~37, no.~2, pp. 156--189, Oct.
  1988. [Online]. Available:
  \url{http://dx.doi.org/10.1016/0022-0000(88)90005-0}
\BIBentrySTDinterwordspacing

\bibitem{Wyner75}
A.~Wyner, ``The wire-tap channel,'' \emph{The Bell Systems Technical Journal},
  vol.~54, no.~8, pp. 1355--1387, Oct. 1975.

\bibitem{IEEEIT:CsiKor78}
I.~Csisz\'ar and J.~K\"{o}rner, ``Broadcast channels with confidential
  messages,'' \emph{Information Theory, IEEE Transactions on}, vol.~24, no.~3,
  pp. 339--348, May 1978.

\bibitem{IEEEIT:Maurer93}
U.~Maurer, ``Secret key agreement by public discussion from common
  information,'' \emph{Information Theory, IEEE Transactions on}, vol.~39,
  no.~3, pp. 733--742, May 1993.

\bibitem{IEEEIT-AhlCsi93}
R.~Ahlswede and I.~Csisz\'ar, ``Common randomness in information theory and
  cryptography. i. secret sharing,'' \emph{Information Theory, IEEE
  Transactions on}, vol.~39, no.~4, pp. 1121--1132, Jul 1993.

\bibitem{FOCS:CreKil88}
C.~Cr{\'e}peau and J.~Kilian, ``Achieving oblivious transfer using weakened
  security assumptions (extended abstract),'' in \emph{29th Annual Symposium on
  Foundations of Computer Science}.\hskip 1em plus 0.5em minus 0.4em\relax
  White Plains, New York: {IEEE} Computer Society Press, Oct.~24--26, 1988, pp.
  42--52.

\bibitem{EC:Crepeau97}
C.~Cr{\'e}peau, ``Efficient cryptographic protocols based on noisy channels,''
  in \emph{Advances in Cryptology -- {EUROCRYPT}'97}, ser. Lecture Notes in
  Computer Science, W.~Fumy, Ed., vol. 1233.\hskip 1em plus 0.5em minus
  0.4em\relax Konstanz, Germany: Springer, Heidelberg, Germany, May~11--15,
  1997, pp. 306--317.

\bibitem{EC:DamKilSal99}
I.~Damg{\aa}rd, J.~Kilian, and L.~Salvail, ``On the (im)possibility of basing
  oblivious transfer and bit commitment on weakened security assumptions,'' in
  \emph{Advances in Cryptology -- {EUROCRYPT}'99}, ser. Lecture Notes in
  Computer Science, J.~Stern, Ed., vol. 1592.\hskip 1em plus 0.5em minus
  0.4em\relax Prague, Czech Republic: Springer, Heidelberg, Germany, May~2--6,
  1999, pp. 56--73.

\bibitem{EC:KhuMajSah16}
D.~Khurana, H.~K. Maji, and A.~Sahai, ``Secure computation from elastic noisy
  channels,'' in \emph{Advances in Cryptology -- {EUROCRYPT}~2016, Part II},
  ser. Lecture Notes in Computer Science, M.~Fischlin and J.-S. Coron, Eds.,
  vol. 9666.\hskip 1em plus 0.5em minus 0.4em\relax Vienna, Austria: Springer,
  Heidelberg, Germany, May~8--12, 2016, pp. 184--212.

\bibitem{cryptoeprint:2016:120}
I.~Cascudo, I.~Damg{\r a}rd, F.~Lacerda, and S.~Ranellucci, ``Oblivious
  transfer from any non-trivial elastic noisy channels via secret key
  agreement,'' Cryptology ePrint Archive, Report 2016/120. To appear on TCC
  2016b, 2016, \url{http://eprint.iacr.org/2016/120}.

\bibitem{WinNasIma03}
\BIBentryALTinterwordspacing
A.~Winter, A.~Nascimento, and H.~Imai,
  ``\BIBforeignlanguage{English}{Commitment capacity of discrete memoryless
  channels},'' in \emph{\BIBforeignlanguage{English}{Cryptography and Coding}},
  ser. Lecture Notes in Computer Science, K.~Paterson, Ed.\hskip 1em plus 0.5em
  minus 0.4em\relax Springer Berlin Heidelberg, 2003, vol. 2898, pp. 35--51.
  [Online]. Available: \url{http://dx.doi.org/10.1007/978-3-540-40974-8_4}
\BIBentrySTDinterwordspacing

\bibitem{IEEEIT:NBSI08}
A.~Nascimento, J.~Barros, S.~Skludarek, and H.~Imai, ``The commitment capacity
  of the gaussian channel is infinite,'' \emph{Information Theory, IEEE
  Transactions on}, vol.~54, no.~6, pp. 2785--2789, June 2008.

\bibitem{IEEEIT:LeuHel78}
S.~Leung-Yan-Cheong and M.~Hellman, ``The gaussian wire-tap channel,''
  \emph{Information Theory, IEEE Transactions on}, vol.~24, no.~4, pp.
  451--456, Jul 1978.

\bibitem{EC:OzaWyn84}
L.~H. Ozarow and A.~D. Wyner, ``Wire-tap channel {II},'' in \emph{Advances in
  Cryptology -- {EUROCRYPT}'84}, ser. Lecture Notes in Computer Science,
  T.~Beth, N.~Cot, and I.~Ingemarsson, Eds., vol. 209.\hskip 1em plus 0.5em
  minus 0.4em\relax Paris, France: Springer, Heidelberg, Germany, Apr.~9--11,
  1985, pp. 33--50.

\bibitem{IEEEIT:CsiNar04}
I.~Csisz\'ar and P.~Narayan, ``Secrecy capacities for multiple terminals,''
  \emph{Information Theory, IEEE Transactions on}, vol.~50, no.~12, pp.
  3047--3061, Dec 2004.

\bibitem{ISIT:ParBla05}
P.~Parada and R.~Blahut, ``Secrecy capacity of simo and slow fading channels,''
  in \emph{Information Theory, 2005. ISIT 2005. Proceedings. International
  Symposium on}, Sept 2005, pp. 2152--2155.

\bibitem{ISIT:BarRod06}
J.~Barros and M.~Rodrigues, ``Secrecy capacity of wireless channels,'' in
  \emph{Information Theory, 2006 IEEE International Symposium on}, July 2006,
  pp. 356--360.

\bibitem{LiYatTra10}
\BIBentryALTinterwordspacing
Z.~Li, R.~Yates, and W.~Trappe, ``\BIBforeignlanguage{English}{Secrecy capacity
  of independent parallel channels},'' in
  \emph{\BIBforeignlanguage{English}{Securing Wireless Communications at the
  Physical Layer}}, R.~Liu and W.~Trappe, Eds.\hskip 1em plus 0.5em minus
  0.4em\relax Springer US, 2010, pp. 1--18. [Online]. Available:
  \url{http://dx.doi.org/10.1007/978-1-4419-1385-2_1}
\BIBentrySTDinterwordspacing

\bibitem{ISIT:LiaPooSha07}
Y.~Liang, H.~Poor, and S.~Shamai, ``Secrecy capacity region of fading broadcast
  channels,'' in \emph{Information Theory, 2007. ISIT 2007. IEEE International
  Symposium on}, June 2007, pp. 1291--1295.

\bibitem{IEEEIT:GopLaiElg08}
P.~K. Gopala, L.~Lai, and H.~El~Gamal, ``On the secrecy capacity of fading
  channels,'' \emph{Information Theory, IEEE Transactions on}, vol.~54, no.~10,
  pp. 4687--4698, Oct 2008.

\bibitem{LaiElgPoo07}
L.~Lai, H.~{El Gamal}, and H.~V. Poor, ``Secrecy capacity of the wiretap
  channel with noisy feedback,'' in \emph{45th Annu. Allerton Conf.
  Communications, Control and Computing}, 2007.

\bibitem{IEEEIT:CsiNar08}
I.~Csisz\'ar and P.~Narayan, ``Secrecy capacities for multiterminal channel
  models,'' \emph{Information Theory, IEEE Transactions on}, vol.~54, no.~6,
  pp. 2437--2452, June 2008.

\bibitem{IEEEIT:AFJK09}
E.~Ardestanizadeh, M.~Franceschetti, T.~Javidi, and Y.-H. Kim, ``Wiretap
  channel with secure rate-limited feedback,'' \emph{Information Theory, IEEE
  Transactions on}, vol.~55, no.~12, pp. 5353--5361, Dec 2009.

\bibitem{BagMotKha09}
G.~Bagherikaram, A.~Motahari, and A.~Khandani, ``Secrecy capacity region of
  gaussian broadcast channel,'' in \emph{Information Sciences and Systems,
  2009. CISS 2009. 43rd Annual Conference on}, March 2009, pp. 152--157.

\bibitem{IEEEIT:EkrUlu11}
E.~Ekrem and S.~Ulukus, ``The secrecy capacity region of the gaussian mimo
  multi-receiver wiretap channel,'' \emph{Information Theory, IEEE Transactions
  on}, vol.~57, no.~4, pp. 2083--2114, April 2011.

\bibitem{IEEEIT:OggHas11}
F.~Oggier and B.~Hassibi, ``The secrecy capacity of the mimo wiretap channel,''
  \emph{Information Theory, IEEE Transactions on}, vol.~57, no.~8, pp.
  4961--4972, Aug 2011.

\bibitem{IEEEIT:NasWin08}
A.~Nascimento and A.~Winter, ``On the oblivious-transfer capacity of noisy
  resources,'' \emph{Information Theory, IEEE Transactions on}, vol.~54, no.~6,
  pp. 2572--2581, June 2008.

\bibitem{ISIT:ImaMorNas06}
H.~Imai, K.~Morozov, and A.~Nascimento, ``On the oblivious transfer capacity of
  the erasure channel,'' in \emph{Information Theory, 2006 IEEE International
  Symposium on}, July 2006, pp. 1428--1431.

\bibitem{ISIT:AhlCsi07}
R.~Ahlswede and I.~Csisz\'ar, ``On oblivious transfer capacity,'' in
  \emph{Information Theory, 2007. ISIT 2007. IEEE International Symposium on},
  June 2007, pp. 2061--2064.

\bibitem{IEEEIT:PDMN11}
A.~Pinto, R.~Dowsley, K.~Morozov, and A.~Nascimento, ``Achieving oblivious
  transfer capacity of generalized erasure channels in the malicious model,''
  \emph{Information Theory, IEEE Transactions on}, vol.~57, no.~8, pp.
  5566--5571, Aug 2011.

\bibitem{DowNas14}
R.~Dowsley and A.~C.~A. Nascimento, ``On the oblivious transfer capacity of
  generalized erasure channels against malicious adversaries: The case of low
  erasure probability,'' \emph{IEEE Transactions on Information Theory},
  vol.~63, no.~10, pp. 6819--6826, Oct 2017.

\bibitem{AC:RenWol05}
R.~Renner and S.~Wolf, ``Simple and tight bounds for information reconciliation
  and privacy amplification,'' in \emph{Advances in Cryptology --
  {ASIACRYPT}~2005}, ser. Lecture Notes in Computer Science, B.~K. Roy, Ed.,
  vol. 3788.\hskip 1em plus 0.5em minus 0.4em\relax Chennai, India: Springer,
  Heidelberg, Germany, Dec.~4--8, 2005, pp. 199--216.

\bibitem{NisZuc96}
\BIBentryALTinterwordspacing
N.~Nisan and D.~Zuckerman, ``Randomness is linear in space,'' \emph{J. Comput.
  Syst. Sci.}, vol.~52, no.~1, pp. 43--52, Feb. 1996. [Online]. Available:
  \url{http://dx.doi.org/10.1006/jcss.1996.0004}
\BIBentrySTDinterwordspacing

\bibitem{DORS08}
\BIBentryALTinterwordspacing
Y.~Dodis, R.~Ostrovsky, L.~Reyzin, and A.~Smith, ``Fuzzy extractors: How to
  generate strong keys from biometrics and other noisy data,'' \emph{SIAM J.
  Comput.}, vol.~38, no.~1, pp. 97--139, Mar. 2008. [Online]. Available:
  \url{http://dx.doi.org/10.1137/060651380}
\BIBentrySTDinterwordspacing

\bibitem{RadTas00}
\BIBentryALTinterwordspacing
J.~Radhakrishnan and A.~Ta-Shma, ``Bounds for dispersers, extractors, and
  depth-two superconcentrators,'' \emph{SIAM J. Discret. Math.}, vol.~13,
  no.~1, pp. 2--24, Jan. 2000. [Online]. Available:
  \url{http://dx.doi.org/10.1137/S0895480197329508}
\BIBentrySTDinterwordspacing

\bibitem{CarWeg77}
\BIBentryALTinterwordspacing
J.~L. Carter and M.~N. Wegman, ``Universal classes of hash functions (extended
  abstract),'' in \emph{Proceedings of the Ninth Annual ACM Symposium on Theory
  of Computing}, ser. STOC '77.\hskip 1em plus 0.5em minus 0.4em\relax New
  York, NY, USA: ACM, 1977, pp. 106--112. [Online]. Available:
  \url{http://doi.acm.org/10.1145/800105.803400}
\BIBentrySTDinterwordspacing

\bibitem{STOC:ImpLevLub89}
R.~Impagliazzo, L.~A. Levin, and M.~Luby, ``Pseudo-random generation from
  one-way functions (extended abstracts),'' in \emph{21st Annual {ACM}
  Symposium on Theory of Computing}.\hskip 1em plus 0.5em minus 0.4em\relax
  Seattle, WA, USA: {ACM} Press, May~15--17, 1989, pp. 12--24.

\bibitem{HILL99}
J.~H{\aa}stad, R.~Impagliazzo, L.~A. Levin, and M.~Luby, ``A pseudorandom
  generator from any one-way function,'' \emph{{SIAM} Journal on Computing},
  vol.~28, no.~4, pp. 1364--1396, 1999.

\bibitem{BenBraRob88}
\BIBentryALTinterwordspacing
C.~H. Bennett, G.~Brassard, and J.-M. Robert, ``Privacy amplification by public
  discussion,'' \emph{SIAM J. Comput.}, vol.~17, no.~2, pp. 210--229, Apr.
  1988. [Online]. Available: \url{http://dx.doi.org/10.1137/0217014}
\BIBentrySTDinterwordspacing

\bibitem{IEEEIT:BBCM95}
C.~Bennett, G.~Brassard, C.~Cr{\'e}peau, and U.~Maurer, ``Generalized privacy
  amplification,'' \emph{Information Theory, IEEE Transactions on}, vol.~41,
  no.~6, pp. 1915--1923, Nov 1995.

\bibitem{Rompel90}
J.~T. Rompel, ``Techniques for computing with low-independence randomness,''
  Ph.D. dissertation, Massachusetts Institute of Technology, Cambridge, MA,
  USA, 1990.

\bibitem{FOCS:Canetti01}
R.~Canetti, ``Universally composable security: A new paradigm for cryptographic
  protocols,'' in \emph{42nd Annual Symposium on Foundations of Computer
  Science}.\hskip 1em plus 0.5em minus 0.4em\relax Las Vegas, NV, USA: {IEEE}
  Computer Society Press, Oct.~14--17, 2001, pp. 136--145.

\bibitem{JIT:DGMN13}
R.~Dowsley, J.~van~de Graaf, J.~M{\"u}ller-Quade, and A.~C.~A. Nascimento, ``On
  the composability of statistically secure bit commitments,'' \emph{Journal of
  Internet Technology}, vol.~14, no.~3, pp. 509--516, 2013.

\bibitem{SBSEG:DMN08}
R.~Dowsley, J.~M{\"u}ller-Quade, and A.~C.~A. Nascimento, ``On the possibility
  of universally composable commitments based on noisy channels,'' in
  \emph{Anais do VIII Simp\'{o}sio Brasileiro em Seguran\c{c}a da
  Informa\c{c}\~{a}o e de Sistemas Computacionais, SBSEG 2008}, A.~L.~M. dos
  Santos and M.~P. Barcellos, Eds.\hskip 1em plus 0.5em minus 0.4em\relax
  Gramado, Brazil: Sociedade Brasileira de Computa\c{c}\~{a}o (SBC), Sep.~1--5,
  2008, pp. 103--114.

\bibitem{IEEEIT:DamPedPfi98}
I.~Damg{\aa}rd, T.~Pedersen, and B.~Pfitzmann, ``Statistical secrecy and
  multibit commitments,'' \emph{Information Theory, IEEE Transactions on},
  vol.~44, no.~3, pp. 1143--1151, May 1998.

\bibitem{TCC:Wullschleger09}
J.~Wullschleger, ``Oblivious transfer from weak noisy channels,'' in
  \emph{TCC~2009: 6th Theory of Cryptography Conference}, ser. Lecture Notes in
  Computer Science, O.~Reingold, Ed., vol. 5444.\hskip 1em plus 0.5em minus
  0.4em\relax Springer, Heidelberg, Germany, Mar.~15--17, 2009, pp. 332--349.

\end{thebibliography}

\end{document}